\documentclass[letterpaper, 10 pt, conference]{ieeeconf}
\IEEEoverridecommandlockouts
\usepackage{cite}
\usepackage{amsmath,amssymb,amsfonts,mathtools}
\usepackage{algorithmic}
\usepackage{graphicx}
\usepackage{textcomp}
\usepackage{bm}
\usepackage{multirow,array}
\newtheorem{prop}{Proposition}
\begin{document}

\title{\LARGE \bf
Game-theoretic Approach to Decision-making Problem for \\ Blockchain Mining
}

\author{Kosuke Toda, Naomi Kuze, \IEEEmembership{Member, IEEE} and Toshimitsu Ushio, \IEEEmembership{Member, IEEE}
\thanks{This work was supported by JST-ERATO HASUO Project Grant Number JPMJER1603, Japan.}
\thanks{The authors are with the Graduate School of Engineering Science, Osaka University, Toyonaka 560-8531, Japan
(email: toda@hopf.sys.es.osaka-u.ac.jp; kuze, ushio@sys.es.osaka-u.ac.jp).
}}

\maketitle
\thispagestyle{empty} 

\begin{abstract}
  It is an important decision-making problem for a miner in the blockchain networks if he/she participates in the mining so that he/she earns a reward by creating a new block earlier than other miners.
  We formulate this decision-making problem as a noncooperative game, because the probability of creating a block depends not only on one's own available computational resources, but also those of other miners.
  Through theoretical and numerical analyses, we show a hysteresis phenomenon of Nash equilibria depending on the reward and a jump phenomenon of miner decisions by a slight change in reward.
  We also show that the reward for which miners decide not to participate in the mining becomes smaller as the number of miners increases.
\end{abstract}


\section{Introduction}
\label{sec:introduction}
Blockchain is a distributed ledger technology for recording transactions that underlies various services such as the digital currency Bitcoin~\cite{nakamoto}.
Blockchain-based services use cryptography to record transactions as a chain of blocks.
A block consists of a block header and transaction data.
The block header contains a cryptographic hash of the previous block, making blockchain-based services resistant to tampering.
In these services, users called \textit{miners} create blocks in a distributed manner, and the longest chain of blocks, called the \textit{main chain}, is considered to be legitimate.
The process of creating blocks is called \textit{mining}.
Blockchain-based services approve transactions through a consensus algorithm,
typically proof-of-work (PoW).
In this algorithm, the mining difficulty is set using a 4-byte value called a \textit{nonce} in the block header.
To create a block, miners must find a nonce such that the cryptographic hash value for the previous block satisfies specific conditions, determined according to the mining difficulty.
In general, a cryptographic hash value for a block is unique according to the nonce contained in the block.
Moreover, a nonce that satisfies the specific conditions cannot be calculated directly.
This is an exhaustive search that imposes a large computational cost on miners.
Consequently, PoW also contributes to the resistance to tampering.
Because transaction approvals depend on miner calculations, miner incentives are important to maintain blockchain-based services.
When a miner successes in creating a block and the created block is contained in the main chain, he/she gets a reward.

Game theory is used to analyze the interaction among rational decision-makers~\cite{game_t}.
A Nash equilibrium is the most accepted solution concept for a non-cooperative game.
Intuitively, it is the most rational solution for all decision-makers in the sense that no decision-maker has a reason for changing his/her decision if the others maintain their decisions.
A non-cooperative game is applied in various fields, such as network security~\cite{network_security} and resource management~\cite{network_communication}.

Because miners make decisions in a distributed and selfish manner, many studies have adopted game theory to analyze their behavior~\cite{survey_game}.
In particular, mining involves significant energy consumption~\cite{bl_energy,mining_cost}, so it is important to analyze miner behavior considering both energy consumption and the expected reward.
Dimitri~\cite{contest} discussed the computational resources needed for mining under a given computation cost and showed that the decision on investment for mining depends only on the average mining cost.
Fiat et~al.~\cite{energy_eq} discussed computational resources spent for mining under an upper limit on time units for hash calculations, showing that all miners use all available resources.
These previous studies implicitly assumed blocks will always be created within a given cost or number of hash calculations.
However, these assumptions are not practical for actual blockchain-based services.

In this paper, we formulate the energy consumption under the condition that miners keep calculating while paying the cost without the upper limit until the block is created. We adopt a deterministic game approach to analyze the relation between mining reward amount and the decision-making regarding participation in mining considering energy consumption and the expected reward.
Specifically, we formulate a utility function according to energy consumption and the mean reward for mining, and model the decision-making problem of miners as a noncooperative game.
Through theoretical and numerical analyses, we show a hysteresis phenomenon of Nash equilibria depending on the reward and a jump phenomenon of miner decisions by a slight change of the reward.
The remainder of this paper is organized as follows.
In Section~\ref{sec:pre}, we formulate the decision-making problem as a noncooperative game.
In Section~\ref{sec:two_miners}, we analyze the Nash equilibrium of the game in the case of two miners.
Section~\ref{sec:numerical_a} presents a numerical analysis.

\section{Game Formulation}\label{sec:pre}
\subsection{Miner decision-making as a game}\label{sec:min}
It is an important decision-making problem for a miner in the blockchain networks if he/she participates in the mining so that he/she earns a reward by creating a new block earlier than other miners.
We formulate this decision-making problem as a non-cooperative game because the probability of creating a block depends not only on one's own available computational resources, but also those of other miners.

We denote $\mathcal{N} = \{1, 2, \ldots, n\}$ as a set of miners ($n \geq 2$) in a blockchain network. Each miner $k \in \mathcal{N}$ has a strategy set $S_k = \{0, 1\}$.
The strategy $s_k = 1$ denotes that miner $k$ participates in the mining and $s_k = 0$ denotes that miner $k$ does not participate in the mining.
Let $s = (s_1, \ldots, s_n)$ and $S = \times_{k \in \mathcal{N}} S_k$ be a strategy profile and the set of strategy profiles, respectively.
We denote $U: S \to \mathbb{R}^n$ as a utility function for all miners and $U_k: S \to \mathbb{R}$ as a utility function for miner $k \in \mathcal{N}$, that is, $U(s) = (U_1(s), \ldots, U_n(s))$ for a given strategy profile $s$.
Then, the game is described as the tuple
\begin{equation}
  G_n = (\mathcal{N}, S, U). \label{eq:n_game_general}
\end{equation}
In Section~\ref{sec:util_def}, we will derive the utility function $U_k$.

\subsection{Derivation of the utility function}\label{sec:util_def}
To create a new block, a miner calculates a hash value $H(\mbox{tx}, \; \mbox{prev.hash}, \; \mbox{nc})$ using the data of the previous block, namely, the Merkle root of transactions $\mbox{tx}$, the hash of previous block header $\mbox{prev.hash}$, and the nonce $\mbox{nc}$.
The hash function~$H$ outputs an $L$-bit hash value ($L \in \mathbb{N}$) according to inputs $\mbox{tx}$, $\mbox{prev.hash}$, and $\mbox{nc}$.
In PoW~\cite{nakamoto,consensus}, the miner needs to find a nonce that satisfies
\begin{equation}
  H(\mbox{tx}, \; \mbox{prev.hash}, \; \mbox{nc}) \leq 2^{L - h}. \label{eq:hash_satisfy}
\end{equation}

For a given target value $2^{L - h}$ in \eqref{eq:hash_satisfy}\footnote{
Note that $h \in \mathbb{N}$. In this context, $h$ corresponds to the difficulty of finding a nonce. The larger $h$ is, the more time is needed for miners to find a nonce. The selection of values $L$ and $h$ depends on the blockchain services.
In Bitcoin~\cite{nakamoto}, $L = 256$ and difficulty $h$ is set so that the average number of generated blocks per hour is constant.
}, the probability that a miner creates a block with one hash calculation is
\begin{equation}
  \mathbb{P}\left[H(\mbox{tx}, \; \mbox{prev.hash}, \; \mbox{nc}) \leq 2^{L - h}\right] = \frac{1}{D}, \nonumber
\end{equation}
where $D = 2^h$~\cite{survey_consensus}.
The relation between blocks and the times they are created is modeled by a Poisson process~\cite{bitcoin_game}.
Let $w_k$ be the average number of queries to $H(\cdot)$ of miner $k \in \mathcal{N}$ calculated per unit operating time.
The rate $\lambda_k$ of the Poisson process for miner $k$ is given by $\lambda_k = w_k / D$~\cite{dif}.

When the miner $k$ participates in the mining, he/she needs a cost $c_k \geq 0$ per unit operating time and calculates queries whose average number per unit operating time depends on the cost, that is, we assume that $w_k = f_k(c_k)$, where $f_k: \mathbb{R}_{+} \to \mathbb{R}_{+}$ is a non-decreasing function.
If miner $k$ chooses $s_k = 1$, then the rate of the Poisson process is
\begin{equation}
  \lambda_k = \frac{s_k f_k(c_k)}{D}. \nonumber
\end{equation}

First, we calculate the expected reward for mining.
Let $\mathcal{M} \subseteq \mathcal{N}$ be the set of miners who choose to participate in the mining.
We assume that all miners in $\mathcal{M}$ start trying to create a new block at time $t = 0$ and that they create same-size blocks. The first miner to create a block that reaches consensus earns a reward $R \geq 0$ ($R \in \mathbb{R}_{+}$)\footnote{The reward includes a fixed reward and a variable one depending on the block size. From the assumption that all miners create same-size blocks, reward $R$ is independent of the miners.}.
Let $B_k(t)$ be the probability of miner $k$ creating a block before other miners be between $t$ and $t + dt$. Then, using the properties of the Poisson process, we have
\begin{align}
  &B_k(t) \nonumber \\
  &=\exp(-\lambda_k t) \lambda_k dt \exp(-\lambda_k dt) \hspace{-2mm} \prod_{i \in \mathcal{M} \setminus \{k\}} \hspace{-2mm} \exp(-\lambda_i (t + dt)) \nonumber \\
  &\approx \lambda_k \exp\left(-\sum_{i \in \mathcal{M}}\lambda_i t\right) dt = \lambda_k \exp\left(-\sum_{i \in \mathcal{N}}\lambda_i t\right) dt. \label{eq:k_t_dt}
\end{align}
From the assumption that all miners create same-size blocks, the probability of earning the reward equals the probability of creating the block~\cite{bitcoin_game}. Then, for any miner $k \in \mathcal{N}$, the probability $P_k(s)$ of earning the reward is given by the integration of \eqref{eq:k_t_dt} in the interval $[0, \infty)$:
\begin{align}
  P_k(s) &= \int_{0}^{\infty} \lambda_k \exp \left(-\sum_{i \in \mathcal{N}} \lambda_i t\right) dt \nonumber \\
  &= \frac{\lambda_k}{\sum_{i \in \mathcal{N}} \lambda_i} = \frac{s_k f_k(c_k)}{\sum_{i \in
  \mathcal{N}} s_i f_i(c_i)}. \nonumber
\end{align}
Note that the probability of miner $k$ earning the reward is $0$ when miner $k$ chooses $s_k = 0$.
Therefore, for any miner $k \in \mathcal{N}$, the expected reward $R_k(s)$ is
\begin{equation}
  R_k(s) = \frac{s_k f_k(c_k)}{\sum_{i \in \mathcal{N}} s_i f_i(c_i)}R. \label{eq:exp_reward}
\end{equation}

Next, we calculate the expected cost for creating a block in the same way as calculating the expected reward.
Assume miner $k$ consumes cost $c_k$ per unit operating time until a nonce that satisfies \eqref{eq:hash_satisfy} is found.
The expected cost for mining $CS_k(s)$ is
\begin{align}
  CS_k(s) &= c_k \int_{0}^{\infty} t \lambda_k \exp\left(- \sum_{i \in \mathcal{N}} \lambda_it \right) dt \nonumber \\
  &= \frac{c_k \lambda_k}{(\sum_{i \in \mathcal{N}}\lambda_i)^2} = \frac{s_k f_k(c_k)}{\sum_{i \in \mathcal{N}}s_i f_i(c_i)}\frac{Dc_k}{\sum_{i \in \mathcal{N}} s_i f_i(c_i)}. \label{eq:cst}
\end{align}

Finally, we define the utility function
\begin{align}
  &U_k(s) \nonumber \\
   &= \begin{cases}
    0 & \mbox{if} \; \; s = (0, \ldots, 0), \\
    \frac{s_k f_k(c_k)}{\sum_{i \in \mathcal{N}}s_i f_i(c_i)} \left(R - \frac{Dc_k}{\sum_{i \in \mathcal{N}}s_i f_i(c_i)}\right) &\mbox{otherwise}
  \end{cases}\label{eq:utility}
\end{align}
as the difference between \eqref{eq:exp_reward} and \eqref{eq:cst}. Note that when no miners decide to participate in the mining, $U_k(s)=0$ for all miners.

\subsection{Extension to mixed strategies}\label{sec:mixed_str}
We denote $X = \times_{k \in \mathcal{N}} X_k$ as a mixed strategy space where
\begin{equation}
  X_k = \{x_k = (x_k^0, 1 - x_k^0)^{\mathrm{T}} \, | \, 0 \leq x_k^0 \leq 1 \}, \nonumber
\end{equation}
and $x_k^0$ represents the probability of miner $k$ choosing $s_k = 0$.
Let $x = (x_1, \ldots, x_n)$ and $x_{-k}$ represent a mixed strategy profile for all miners and for all miners except $k$, respectively.
We define $e_m^i$ as an $m$-dimensional unit vector in which only the $i$-th component is 1.
Note that mixed strategy profiles $(e_2^1, \ldots, e_2^1)$ and $(e_2^2, \ldots, e_2^2)$ correspond to pure strategy profiles $s = (0, \ldots, 0)$ and $s = (1, \ldots, 1)$, respectively.
Then, for any mixed strategy profile $x$, the expected utility function $u_k: X \to \mathbb{R}$ of miner $k \in \mathcal{N}$ is
\begin{equation}
  u_k(x) = x_k^0 u_k(e_2^1, x_{-k}) + (1 - x_k^0) u_k(e_2^2, x_{-k}). \nonumber
\end{equation}
Note that $u_k(e_2^1, x_{-k})$ and $u_k(e_2^2, x_{-k})$ are the expected utility values when miner $k$ chooses strategy $s_k = 0$ and $s_k = 1$, respectively.

Let $\tilde{\beta}_k : X \to 2^{X_k}$ represent the best response correspondence of miner $k$ for mixed strategy $x$ as
\begin{align}
  &\tilde{\beta}_k(x) \nonumber \\
  &=\{x_k^{*} \in X_k \, | \, \forall x_k^{\prime} \in X_k \; u_k(x_k^{*}, x_{-k}) \geq u_k(x_k^{\prime}, x_{-k})\}. \nonumber
\end{align}
A mixed strategy $x^{*}$ satisfying $x^{*} \in \tilde{\beta}(x^{*})$ is called a Nash equilibrium, where $\tilde{\beta}(x) = \times_{k \in \mathcal{N}} \tilde{\beta}_k(x) \subseteq X$. We denote $\mbox{NE}(G_n)$ as the set of Nash equilibria in game $G_n$.

\section{The two-miners case}\label{sec:two_miners}
In this section, we focus on the case of two miners, namely, where $\mathcal{N} = \{1, 2\}$ and $f_k(c_k)$ is a linear function $f_k(c_k) = v_k c_k, \, v_k \in \mathbb{R}_{+}$ for all $k \in \mathcal{N}$.
We theoretically derive the Nash equilibria of game $G_2$, where
the utility functions of two miners are written as
\begin{align}
  &U_1(s) \nonumber \\
  &= \begin{cases}
    0 & \mbox{if} \; \; (s_1, s_2) = (0, 0), \\
    \frac{s_1}{s_1 + s_2 p_v p_c}\left(R - \frac{d}{s_1 + s_2 p_v p_c}\right) & \mbox{otherwise},
\end{cases}\label{eq:util_m1} \\
  &U_2(s) \nonumber \\
  &= \begin{cases}
  0 & \mbox{if} \; \; (s_1, s_2) = (0, 0), \\
  \frac{s_2 p_v p_c}{s_1 + s_2 p_v p_c}\left(R - \frac{dp_c}{s_1 + s_2 p_v p_c}\right) & \mbox{otherwise},
\end{cases}\label{eq:util_m2}
\end{align}
where $p_v \coloneqq v_2 / v_1 > 0$, $p_c \coloneqq c_2 / c_1 > 0$, and $d \coloneqq D / v_1 > 0$.
We can assume $p_v \geq 1 \; (0 < v_1 \leq v_2)$ without loss of generality\footnote{Intuitively, $v_k$ is the cost-effectiveness, and
$p_v$ and $p_c$ are parameters that represent ratios of cost-effectiveness and cost, respectively.
When $v_1$ is constant, the larger the difficulty $D$ (difficulty level $h$), the larger the value of $d$.}.

Strategic forms for finite two-player games are depicted as matrices. Tables~\ref{tab:pr_1} and \ref{tab:pr_2} show payoffs for the corresponding strategy profiles of miners~1 and 2, respectively.
\begin{table}[b]
  \centering
  \caption{Miner 1's payoff for the corresponding strategy profile.}
  \setlength{\extrarowheight}{2pt}
  \begin{tabular}{cc|c|c|}
    & \multicolumn{1}{c}{} & \multicolumn{2}{c}{$S_2$}\\
    & \multicolumn{1}{c}{} & \multicolumn{1}{c}{$0$}  & \multicolumn{1}{c}{$1$} \\\cline{3-4}
    \multirow{2}*{$S_1$}  & $0$ & $0$ & $0$ \\\cline{3-4}
    & $1$ & $R - d$ & $\displaystyle \frac{1}{1 + p_v p_c}\left(R - \frac{d}{1 + p_v p_c}\right)$ \\\cline{3-4}
  \end{tabular}
  \label{tab:pr_1}
\end{table}
\begin{table}[b]
  \centering
  \caption{Miner 2's payoff for the corresponding strategy profile.}
  \setlength{\extrarowheight}{2pt}
  \begin{tabular}{cc|c|c|}
    & \multicolumn{1}{c}{} & \multicolumn{2}{c}{$S_2$}\\
    & \multicolumn{1}{c}{} & \multicolumn{1}{c}{$0$}  & \multicolumn{1}{c}{$1$} \\\cline{3-4}
    \multirow{2}*{$S_1$}  & $0$ & $0$ & $\displaystyle R - \frac{d}{p_v}$ \\\cline{3-4}
    & $1$ & $0$ & $\displaystyle \frac{p_v p_c}{1 + p_v p_c}\left(R - \frac{dp_c}{1 + p_v p_c}\right)$ \\\cline{3-4}
  \end{tabular}
  \label{tab:pr_2}
\end{table}
We obtain the set of Nash equilibria for mixed strategies of game $G_2$ as in Proposition~\ref{prop:nash}.

\begin{prop}\label{prop:nash}
  Assume $p_v \geq 1$.
  We define functions $g_1: \mathbb{R} \to \mathbb{R}$ and $g_2: \mathbb{R} \to \mathbb{R}$ as
  \begin{align}
    &g_1(z) = p_v p_c \left(z - \frac{p_c}{1 + p_v p_c}\right) \left(\frac{2p_v p_c + 1}{p_v(1 + p_v p_c)} - z\right)^{-1}, \label{eq:gz1} \\
    &g_2(z) = \left(z - \frac{1}{1 + p_v p_c}\right)\left(p_v p_c \left(\frac{2 + p_v p_c}{1 + p_v p_c} - z\right)\right)^{-1}. \label{eq:gz2}
  \end{align}
  Respectively letting $\alpha_1$ and $\alpha_2$ be $\alpha_1 \coloneqq g_1(R/d)$ and $\alpha_2 \coloneqq g_2(R/d)$,
  the set $\mbox{NE}(G_2)$ is given as follows\footnote{
  Intuitively, the value $R/d$ represents how much reward is given for the difficulty,
  increasing with reward $R$ and decreasing with difficulty level $h$.}:

  If $p_c \geq 1$, then
  \begin{align}
    &\mbox{NE}(G_2) = \nonumber \\
    &\begin{cases}
      \{(e_2^1, e_2^1)\} \\
      \hspace{15mm} \mbox{if} \; \; \frac{R}{d} < \frac{p_c}{1 + p_v p_c}, \\
      \{(e_2^1, e_2^1), (e_2^2, (\gamma_2, 1 - \gamma_2)^{\mathrm{T}})\} \\
      \hspace{15mm} \mbox{if} \; \; \frac{R}{d} = \frac{p_c}{1 + p_v p_c}, \\
      \{(e_2^1, e_2^1), (e_2^2, e_2^2), ((\alpha_1, 1 - \alpha_1)^{\mathrm{T}}, (\alpha_2, 1 - \alpha_2)^{\mathrm{T}})\} \\
      \hspace{15mm} \mbox{if} \; \; \frac{p_c}{1 + p_v p_c} < \frac{R}{d} < \frac{1}{p_v}, \\
      \{(e_2^2, e_2^2), (e_2^1, (\delta_2, 1 - \delta_2)^{\mathrm{T}})\} \\
      \hspace{15mm} \mbox{if} \; \; \frac{R}{d} = \frac{1}{p_v}, \\
      \{(e_2^2, e_2^2)\} \\
      \hspace{15mm} \mbox{if} \; \; \frac{R}{d} > \frac{1}{p_v},
    \end{cases}\label{eq:NEG2_1}
  \end{align}
  where $\gamma_2 \in [0, g_2(p_c / (1 + p_v p_c))]$ and $\delta_2 \in [g_2(1 / p_v), 1]$.

  If instead $1 - 1 / p_v \leq p_c < 1$, then
  \begin{align}
    &\mbox{NE}(G_2) = \nonumber \\
    &\begin{cases}
      \{(e_2^1, e_2^1)\} \\
      \hspace{15mm} \mbox{if} \; \; \frac{R}{d} < \frac{1}{1 + p_v p_c}, \\
      \{(e_2^1, e_2^1), ((\varepsilon_1, 1 - \varepsilon_1)^{\mathrm{T}}, e_2^2)\}, \\
      \hspace{15mm} \mbox{if} \; \; \frac{R}{d} = \frac{1}{1 + p_v p_c}, \\
      \{(e_2^1, e_2^1), (e_2^2, e_2^2), ((\alpha_1, 1 - \alpha_1)^{\mathrm{T}}, (\alpha_2, 1 - \alpha_2)^{\mathrm{T}})\} \\
      \hspace{15mm} \mbox{if} \; \; \frac{1}{1 + p_v p_c} < \frac{R}{d} < \frac{1}{p_v}, \\
      \{(e_2^2, e_2^2), (e_2^1, (\delta_2, 1 - \delta_2)^{\mathrm{T}})\} \\
      \hspace{15mm} \mbox{if} \; \; \frac{R}{d} = \frac{1}{p_v}, \\
      \{(e_2^2, e_2^2)\} \\
      \hspace{15mm} \mbox{if} \; \; \frac{R}{d} > \frac{1}{p_v},
    \end{cases}\label{eq:NEG2_2}
  \end{align}
  where $\varepsilon_1 \in [0, g_1(1 / (1 + p_v p_c))]$ and $\delta_2 \in [g_2(1 / p_v), 1]$.

 Finally, if $0 < p_c < 1 - 1 / p_v$, then
  \begin{align}
    \mbox{NE}(G_2) = \begin{cases}
      \{(e_2^1, e_2^1)\} & \mbox{if} \; \; \frac{R}{d} < \frac{1}{p_v}, \\
      \{(e_2^1, (\zeta_2, 1 - \zeta_2)^{\mathrm{T}})\} & \mbox{if} \; \; \frac{R}{d} = \frac{1}{p_v}, \\
      \{(e_2^1, e_2^2)\} & \mbox{if} \; \; \frac{1}{p_v} < \frac{R}{d} < \frac{1}{1 + p_v p_c}, \\
      \{((\eta_1, 1 - \eta_1)^{\mathrm{T}}, e_2^2)\} & \mbox{if} \; \; \frac{R}{d} = \frac{1}{1 + p_c p_v}, \\
      \{(e_2^2, e_2^2)\} & \mbox{if} \; \; \frac{R}{d} > \frac{1}{1 + p_v p_c},
    \end{cases}\label{eq:NEG2_3}
  \end{align}
  where $\zeta_2 \in [0, 1]$ and $\eta_1 \in [0, 1]$.
\end{prop}

\begin{proof}
  Because the pure strategy sets $S_1$ and $S_2$ are finite, there exists at least one Nash equilibrium in game $G_2$ in the mixed strategies~\cite{game_t}.
  $x_1 = (x_1^0, 1 - x_1^0)^{\mathrm{T}}$, $0 \leq x_1^0 \leq 1$ and $x_2 = (x_2^0, 1 - x_2^0)^{\mathrm{T}}$, $0 \leq x_2^0 \leq 1$ represent mixed strategies of miners~1 and 2, respectively.
  We denote the expected payoffs of miners~1 and 2 as $u_1(x_1, x_2)$ and $u_2(x_1, x_2)$, respectively.
  The difference in expected payoffs between $s_k = 1$ and $s_k = 0$ for miner $k$ is
  \begin{align}
    &u_1(e_2^2, x_2) - u_1(e_2^1, x_2) \nonumber \\
    &= \frac{d}{1 + p_v p_c}\left(
    \left(\frac{R}{d} - \frac{1}{1 + p_v p_c}\right)\right. \nonumber \\
    &\hspace{25mm}+ \left. p_vp_c\left(\frac{R}{d} - \frac{2 + p_v p_c}{1 + p_v p_c}\right)x_2^0 \right) \label{eq:dif_u1}
  \end{align}
  and
  \begin{align}
    &u_2(x_1, e_2^2) - u_2(x_1, e_2^1) \nonumber \\
    &= \frac{d}{1 + p_v p_c}\left(
    p_v p_c \left(\frac{R}{d} - \frac{p_c}{1 + p_v p_c}\right)\right. \nonumber \\
    &\hspace{25mm}+ \left. \left(\frac{R}{d} - \frac{2p_v p_c + 1}{p_v (1 + p_v p_c)}\right)x_1^0 \right). \label{eq:dif_u2}
  \end{align}
  From \eqref{eq:dif_u1} and \eqref{eq:dif_u2}, for another miner $\ell \in \mathcal{N} \setminus \{k\}$, the best response $\tilde{\beta}_k(x)$ of miner $k$ changes depending on the value range for $g_{\ell}(R / d)$.
  We thus consider five cases: $g_{\ell}(R / d) < 0$, $g_{\ell}(R / d) = 0$, $0 < g_{\ell}(R / d) < 1$, $g_{\ell}(R / d) = 1$, and $g_{\ell}(R / d) > 1$.

  When $0 < g_{\ell}(R/d) < 1$, the sign of the difference between miner $k$'s expect payoff for $s_k = 1$ and $s_k = 0$ changes depending on the value of $x_{\ell}^0$ as follows\footnote{We present the case where $k = 1$ and $\ell = 2$, but the same applies when $k = 2$ and $\ell = 1$.}.
  \begin{equation}
    \begin{cases}
      u_k(e_2^2, x_{\ell}) - u_k(e_2^1, x_{\ell}) < 0 & \mbox{if} \; \; g_{\ell}(R/d) < x_{\ell}^0 \leq 1, \\
      u_k(e_2^2, x_{\ell}) - u_k(e_2^1, x_{\ell}) = 0 & \mbox{if} \; \; x_{\ell}^0 = g_{\ell}(R/d), \\
      u_k(e_2^2, x_{\ell}) - u_k(e_2^1, x_{\ell}) > 0 & \mbox{if} \; \; 0 \leq x_{\ell}^0 < g_{\ell}(R/d).
    \end{cases} \nonumber
  \end{equation}
  Therefore, the best response $\tilde{\beta}_{k}(x)$ is
  \begin{equation}
    \tilde{\beta}_k(x) = \begin{cases}
      \{e_2^1\} & \mbox{if} \; \; g_{\ell}(R / d) < x_{\ell}^0 \leq 1, \\
      X_k & \mbox{if} \; \; x_{\ell}^0 = g_{\ell}(R/d), \\
      \{e_2^2\} & \mbox{if} \; \; 0 \leq x_{\ell}^0 < g_{\ell}(R / d).
  \end{cases}\label{eq:bt_pat3}
  \end{equation}
  Similarly, we obtain the best response $\tilde{\beta}_{k}(x)$ depending on the value range for $g_{\ell}(R/d)$ as follows:
  \begin{itemize}
    \item When $g_{\ell}(R / d) < 0$,
    \begin{equation}
      \tilde{\beta}_k(x) = \{e_2^1\}. \label{eq:bt_pat1}
    \end{equation}
    \item When $g_{\ell}(R / d) = 0$,
    \begin{equation}
      \tilde{\beta}_k(x) = \begin{cases}
        X_k & \mbox{if} \; \; x_{\ell}^0 = 0, \\
        \{e_2^1\} & \mbox{if} \; \; 0 < x_{\ell}^0 \leq 1.
    \end{cases}\label{eq:bt_pat2}
    \end{equation}
    \item When $g_{\ell}(R / d) = 1$,
    \begin{equation}
      \tilde{\beta}_k(x) = \begin{cases}
        X_k & \mbox{if} \; \; x_{\ell}^0 = 1, \\
        \{e_2^2\} & \mbox{if} \; \; 0 \leq x_{\ell}^0 < 1.
    \end{cases}\label{eq:bt_pat4}
    \end{equation}
    \item When $g_{\ell}(R / d) > 1$,
    \begin{equation}
      \tilde{\beta}_k(x) = \{e_2^2\}. \label{eq:bt_pat5}
    \end{equation}
  \end{itemize}

  We next derive the range of $R/d$ satisfying $0 < g_k(R/d) < 1, \, k = 1, 2$.
  From \eqref{eq:gz1} and \eqref{eq:gz2}, it is easily shown that $g_1$ and $g_2$ are monotonically increasing functions. Therefore, we obtain
  \begin{align}
    &0 < g_1(R / d) < 1 \Rightarrow \frac{p_c}{1 + p_v p_c} < \frac{R}{d} < \frac{1}{p_v}, \label{eq:g1_range} \\
    &0 < g_2(R/d) < 1 \Rightarrow \frac{1}{1 + p_v p_c} < \frac{R}{d} < 1. \label{eq:g2_range}
  \end{align}
  Assuming $p_v \geq 1$, the magnitude relationship among $p_c/(1 + p_v p_c)$, $1 / p_v$, $1 / (1 + p_v p_c)$, and $1$ depends on the value of $p_c$, as follows:
  \begin{enumerate}
    \item When $p_c \geq 1$, \label{enu:case1}
    \begin{equation}
      \frac{1}{1 + p_v p_c} \leq \frac{p_c}{1 + p_v p_c} < \frac{1}{p_v} \leq 1. \nonumber
    \end{equation}
    \item When $1 - 1/p_v \leq p_c < 1$, \label{enu:case2}
    \begin{equation}
      \frac{p_c}{1 + p_v p_c} < \frac{1}{1 + p_v p_c} \leq \frac{1}{p_v} \leq 1. \nonumber
    \end{equation}
    \item When $0 < p_c < 1 - 1/p_v$, \label{enu:case3}
    \begin{equation}
      \frac{p_c}{1 + p_v p_c} < \frac{1}{p_v} < \frac{1}{1 + p_v p_c} < 1. \nonumber
    \end{equation}
  \end{enumerate}

  We derive the set of Nash equilibria in the case of \ref{enu:case1}).
  Then, we have nine cases depending on the value range of $R/d$.
  Table~\ref{tab:rdnash} shows the relation between $R/d$ and the set of Nash equilibria.
  Therefore, \eqref{eq:NEG2_1} is the set of Nash equilibria.
  We can similarly prove \eqref{eq:NEG2_2} and \eqref{eq:NEG2_3}.
\end{proof}

\begin{table*}[t]
  \centering
  \caption{Relation between $R/d$ and the set of Nash equilibria.}
  \begin{tabular}{|c|c|c|c|}\hline
    $R/d$ value & $\alpha_1$ & $\alpha_2$ & Set of Nash equilibria \\ \hline
    $0 < R/d < 1 / (1 + p_v p_c)$ & $\alpha_1 < 0$ & $\alpha_2 < 0$ & $\{(e_2^1, e_2^1)\}$ \\
    $R/d = 1 / (1 + p_v p_c)$ & $\alpha_1 < 0$ & $\alpha_2 = 0$ & $\{(e_2^1, e_2^1)\}$ \\
    $1 / (1 + p_v p_c) < R / d < p_c / (1 + p_v p_c)$ & $\alpha_1 < 0$ & $0 < \alpha_2 < 1$ & $\{(e_2^1, e_2^1)\}$ \\
    $R/d = p_c / (1 + p_v p_c)$ & $\alpha_1 = 0$ & $0 < \alpha_2 < 1$ & $\{(e_2^1, e_2^1), (e_2^2, (\gamma_2, 1 - \gamma_2)^{\mathrm{T}})\}, \; \; \gamma_2 \in [0, g_2(p_c/(1 + p_v p_c))]$ \\
    $p_c / (1 + p_v p_c) < R / d < 1 / p_v$ & $0 < \alpha_1 < 1$ & $0 < \alpha_2 < 1$ & $\{(e_2^1, e_2^1), (e_2^2, e_2^2), ((\alpha_1, 1 - \alpha_1)^{\mathrm{T}}, (\alpha_2, 1 - \alpha_2)^{\mathrm{T}})\}$ \\
    $R/d = 1 / p_v$ & $\alpha_1 = 1$ & $0 < \alpha_2 < 1$ & $\{(e_2^2, e_2^2), (e_2^1, (\delta_2, 1 - \delta_2)^{\mathrm{T}})\}, \; \; \delta_2 \in [g_2(1 / p_v), 1]$ \\
    $1 / p_v < R / d < 1$ & $\alpha_1 > 1$ & $0 < \alpha_2 < 1$ & $\{(e_2^2, e_2^2)\}$ \\
    $R/d = 1$ & $\alpha_1 > 1$ & $\alpha_2 = 1$ & $\{(e_2^2, e_2^2)\}$ \\
    $R/d > 1$ & $\alpha_1 > 1$ & $\alpha_2 > 1$ & $\{(e_2^2, e_2^2)\}$ \\ \hline
  \end{tabular}
  \label{tab:rdnash}
\end{table*}

For a given $p_v$, $p_c$ affects change in the Nash equilibria depending on $R/d$.
Fig.~\ref{fig:pcR} shows the $p_c-R/d$ parameter plane where $p_v \geq 1$ is fixed.
If the pair $(p_c, R/d)$ is in region~(a), (b), (c), or (d), set $\mbox{NE}(G_2)$ satisfies
$\mbox{NE}(G_2) = \{(e_2^2, e_2^2)\}$,
$\mbox{NE}(G_2) = \{(e_2^1, e_2^1), (e_2^2, e_2^2), ((\alpha_1, 1 - \alpha_1)^{\mathrm{T}}, (\alpha_2, 1 - \alpha_2)^{\mathrm{T}})\}$,
$\mbox{NE}(G_2) = \{(e_2^1, e_2^1)\}$, or
$\mbox{NE}(G_2) = \{(e_2^1, e_2^2)\}$, respectively\footnote{$\mbox{NE}(G_2)$ for the boundary between regions~(a) and (b) is given by the fourth equation in \eqref{eq:NEG2_1} and the fourth equation in \eqref{eq:NEG2_2}.
Similarly, $\mbox{NE}(G_2)$ for the boundaries of regions~(b) and (c), (c) and (d), and (d) and (a) are given by the second equation in \eqref{eq:NEG2_1} if $p_c \geq 1$ and the second equation in \eqref{eq:NEG2_2} if $1 - 1/p_v < p_c < 1$, the second equation in \eqref{eq:NEG2_3}, and the fourth equation in \eqref{eq:NEG2_3}, respectively.
$\mbox{NE}(G_2)$ when the pair $(p_c, R/d) = (1 - 1/p_v, 1/p_v)$ is given by $\mbox{NE}(G_2) = \{((\varepsilon_1, 1 - \varepsilon_1)^{\mathrm{T}}, e_2^2), (e_2^1, (\delta_2, 1 - \delta_2)^{\mathrm{T}})\}$,
where $\varepsilon_1 \in [0, 1]$ and $\delta_2 \in [0, 1]$.}.

\begin{figure}[tb]
  \centering
  \includegraphics[clip, width = 7.5cm]{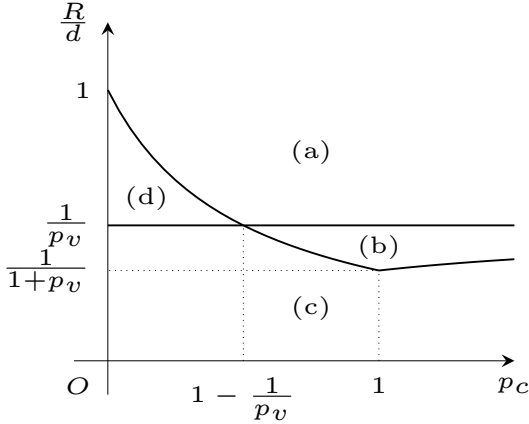}
  \caption{The $p_c-R/d$ parameter plane where $p_v \geq 1$ is fixed.}
  \label{fig:pcR}
\end{figure}

Figs.~\ref{fig:eq_1}--\ref{fig:eq_3} show the relation between $R/d$ and $x_k^0$ ($k = 1, 2$) in Nash equilibria.
The blue (left) and red (right) lines represent values for $x_1^0$ and $x_2^0$, respectively.

\begin{figure}[tb]
  \centering
  \includegraphics[clip, width = 7.5cm]{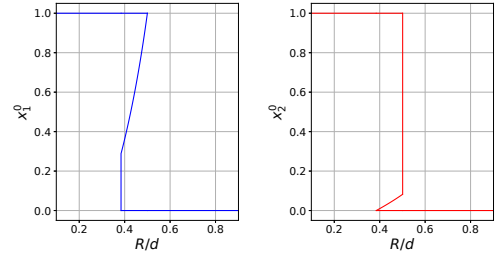}
  \caption{The case when $(p_v, p_c) = (2, 0.8)$.}
  \label{fig:eq_1}
\end{figure}

Fig.~\ref{fig:eq_1} shows the case where $p_c > 1 - 1/p_v$ is fixed.
Consider the case where $1 - 1/p_v < p_c < 1$.
The pure strategy profile $s = (0, 0)$ is a Nash equilibrium if $R/d$ is smaller than $1/p_v$.
If $R/d$ is larger than $1/p_v$, the pure Nash equilibrium $s = (0, 0)$ disappears and only the pure strategy profile $s = (1, 1)$ is a Nash equilibrium.

The pure strategy profile $s = (1, 1)$ is a Nash equilibrium if $R/d$ is larger than $1/(1 + p_v p_c)$.
If $R/d$ is smaller than $1/(1 + p_v p_c)$, the pure Nash equilibrium $s = (1, 1)$ disappears and only the pure strategy profile $s = (0, 0)$ is a Nash equilibrium.

This change in Nash equilibria due to change in $R/d$ implies that when the mining reward exceeds some value, all miners will decide to participate in the mining, after which they continue for a while even if the reward decreases to the boundary of region~(b). Thus, a hysteresis phenomenon exists in the region.
Moreover, a jump phenomenon regarding the strategy profiles miners choose occurs owing to the disappearance of Nash equilibria when the reward changes across the region boundary.

\begin{figure}[tb]
  \centering
  \includegraphics[clip, width = 7.5cm]{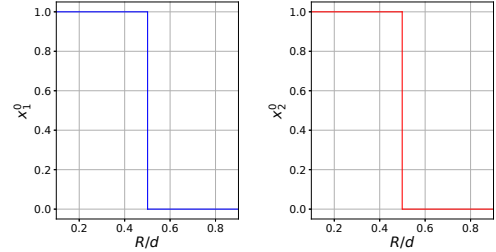}
  \caption{The case when $(p_v, p_c) = (2, 0.5)$.}
  \label{fig:eq_2}
\end{figure}

Fig.~\ref{fig:eq_2} shows the case where $p_c = 1 - 1/p_v$.
A transition from region~(c) to (a), that is, a jump from one pure strategy profile to another, is observed when $R/d$ equals $1/p_v = 1/(1 + p_v p_c)$.
This transition is only seen when $p_c = 1 - 1 / p_v$.

\begin{figure}[tb]
  \centering
  \includegraphics[clip, width = 7.5cm]{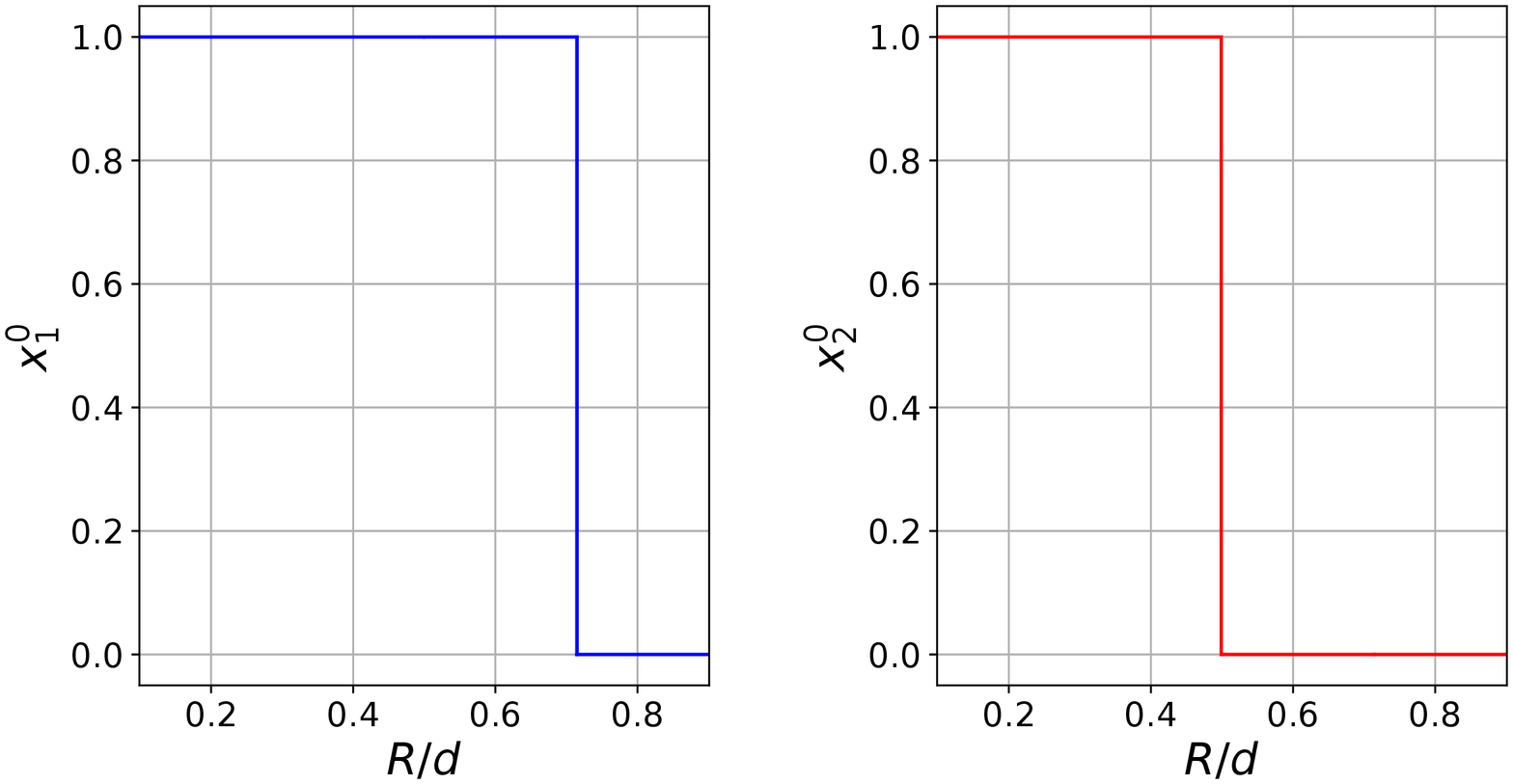}
  \caption{The case when $(p_v, p_c) = (2, 0.2)$.}
  \label{fig:eq_3}
\end{figure}

Fig.~\ref{fig:eq_3} shows the case where $p_c < 1 - 1/p_v$ is fixed.
The pure strategy profile $s = (0, 1)$ is a Nash equilibrium if $1/p_v < R/d < 1/(1 + p_v p_c)$.
Two pure Nash equilibria do not coexist in the interior of each region while mixed strategy profile $(e_2^1, (\zeta_2, 1 - \zeta_2)^{\mathrm{T}})$ and $((\eta_1, 1 - \eta_1)^{\mathrm{T}}, e_2^2)$ are Nash equilibria on the boundary $R/d = 1/p_v$ and $R/d = 1/(1 + p_v p_c)$.

\section{Numerical Analysis}\label{sec:numerical_a}
In this section, we use Gambit~\cite{gambit} to numerically calculate Nash equilibria for the game $G_n$ with $n > 2$.
We estimate $c_k = c$ and $f_k(c) = vc, \, v > 0$ for all $k \in \mathcal{N}$\footnote{In this example, all miners calculate the same number of hash queries per unit operating time.}. From \eqref{eq:utility} and these assumptions, $U_k(s)$ can be rewritten as
\begin{equation}
  U_k(s) = \begin{cases}
    0 & \mbox{if} \; \; s = (0, \ldots, 0), \\
    \frac{s_k}{\sum_{i \in \mathcal{N}} s_i} \left(R - \frac{d}{\sum_{i \in \mathcal{N}} s_i}\right) & \mbox{otherwise},
\end{cases}\nonumber
\end{equation}
where $d \coloneqq D/v$. Value $d$ is fixed to 100 and set $R \in \{0, 1, \ldots, 150\}$. We calculate Nash equilibria for each $R$ when the number of miners $n$ is 2, 3, 4, 5, and 6.

\begin{figure}[tb]
  \centering
  \includegraphics[clip, width = 7.5cm]{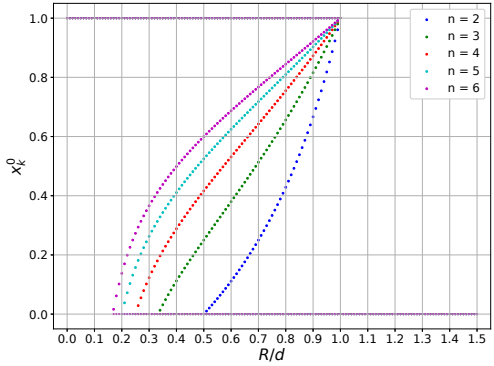}
  \caption{Relation between $R/d$ and $x_k^0$ in Nash equilibria for each the number of miners $n$.}
  \label{fig:num_1}
\end{figure}

Fig.~\ref{fig:num_1} shows the relation between $R/d$ and $x_k^0$ in Nash equilibria for numbers of miners $n = 2, 3, 4, 5, 6$.
This result shows that both the hysteresis phenomena and the jump phenomenon of the strategy profiles can be observed regardless of the number of miners when all miners pay the same cost and calculate the same number of hash queries per unit operating time.
In addition, this implies that as the number of miners increases, $R/d$ for the appearance of a Nash equilibrium $s = (1, \ldots, 1)$ decreases.

\begin{figure}[tb]
  \centering
  \includegraphics[clip, width = 7.5cm]{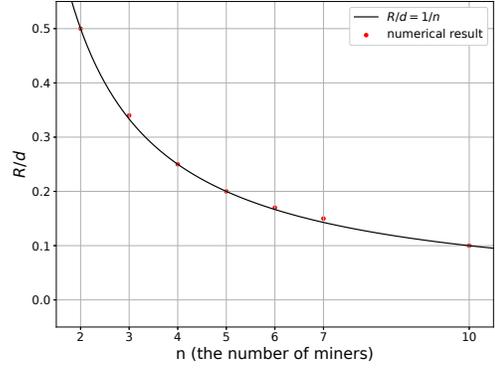}
  \caption{Relation between number of miners and $R/d$ when the Nash equilibrium $s = (1, \ldots, 1)$ appears.}
  \label{fig:num_2}
\end{figure}

Fig.~\ref{fig:num_2} shows the relation between the number of miners and $R/d$ when the Nash equilibrium $s = (1, \ldots, 1)$ appears. The figure shows that $R/d$ when the equilibrium $s = (1, \ldots, 1)$ appears is inversely proportional to the number of miners $n$.

This result shows that once miners decide to participate in the mining, they continue to mine for smaller rewards as their number increases.
Thus, it is very important when designing blockchain networks to set the largest possible initial reward as an incentive for mining in the network.
The reward can later be decreased, after the number of participating miners increases, without decreasing their number.

\section{Conclusion}\label{sec:conclusion}
Modeling a decision-making problem for mining participation as a noncooperative game, we showed that hysteresis phenomena due to the coexistence of two pure Nash equilibria and jump phenomena in the choice of strategy profiles can be observed with changes in the mining reward.
Moreover, numerical calculations showed that miners continue mining for smaller rewards as their number increases.
In general, it is difficult to analyze the miner behavior as $n$ increases.
In future work we will theoretically analyze Nash equilibria by deriving a macro model of the group of the miners.

\end{document}